\documentclass[12pt]{article}
\usepackage{sectsty}
\usepackage[T1]{fontenc}
\usepackage{kpfonts}
\usepackage[dvipsnames]{xcolor}
\usepackage{pdfpages}
\usepackage{sgame}
\usepackage{accents}
\usepackage{tikz-cd}
\usepackage{float}
\usepackage[round]{natbib}
\usepackage{multirow}
\usepackage{dutchcal}
\usepackage{pdfpages}
\usepackage{bbm}
\usepackage{sgame}
\usepackage{mathtools}
\usepackage{amsthm}
\usepackage{enumitem}
\usepackage[margin=1.25in]{geometry}
\usepackage{caption}
\usepackage{subcaption}
\usepackage{epigraph}
\usepackage{listings}
\usetikzlibrary{calc}
\usetikzlibrary{shapes,arrows}
\usepackage{tcolorbox}

\newcommand{\ubar}[1]{\underaccent{\bar}{#1}}

\DeclareMathOperator\co{co}

\newtheorem{theorem}{Theorem}[section]
\newtheorem{proposition}[theorem]{Proposition}
\newtheorem{lemma}[theorem]{Lemma}
\newtheorem{corollary}[theorem]{Corollary}

\theoremstyle{definition}

\newtheorem{definition}[theorem]{Definition}

\newtheorem{remark}[theorem]{Remark}

\sectionfont{\color{BlueViolet}}
\subsectionfont{\color{BlueViolet}}

\usepackage{hyperref}
\hypersetup{
    colorlinks=true,
    linkcolor=CadetBlue,
    filecolor=CadetBlue,      
    urlcolor=CadetBlue,
    citecolor=Mulberry,
}

\usepackage{setspace}

\definecolor{backcolour}{rgb}{0.63, 0.79, 0.95}

\lstdefinestyle{mystyle}{
  backgroundcolor=\color{backcolour},
  basicstyle=\ttfamily\footnotesize,
  breakatwhitespace=false,         
  breaklines=true,                 
  captionpos=b,                    
  keepspaces=true,                 
  numbers=left,                    
  numbersep=5pt,                  
  showspaces=false,                
  showstringspaces=false,
  showtabs=false,                  
  tabsize=2
}

\providecommand{\keywords}[1]{\textbf{\textit{Keywords:}} #1}
\providecommand{\jel}[1]{\textbf{\textit{JEL Classifications:}} #1}

\begin{document}
\author{Marilyn Pease \and Mark Whitmeyer\thanks{MP: Kelley School of Business, Indiana University, \href{marpease@iu.edu}{marpease@iu.edu} \& MW: Arizona State University, \href{mailto:mark.whitmeyer@gmail.com}{mark.whitmeyer@gmail.com}. This paper was previously titled ``Safety, in Numbers.'' We thank Brian Albrecht, 
Tilman B\"{o}rgers, Hector Chade, Francesco Fabbri, B\r{a}rd Harstad, Shuo Liu, Kevin Reffett, Eddie Schlee, Alex Teytelboym, Joseph Whitmeyer, Tom Wiseman, Kun Zhang, and conference and seminar audiences at Duke/UNC, Durham University, Georgia Tech, Texas A\&M, and the University of Queensland. Konstantin von Beringe provided excellent research assistance. Pease thanks the John Rau Kelley School of Business Faculty Fellowship.}}

\title{Playing It Safe: Actions Attractive to the Risk Averse}
\maketitle

\begin{abstract}
We introduce a way to compare actions in decision problems. One action is \textit{safer} than another if the set of beliefs at which the decision-maker prefers the safer action expands as the decision-maker becomes more risk averse. We provide a full characterization of this relation, show that it is equivalent to robust conceptions of single-crossing and second-order stochastic dominance, and reveal that in monotone decision problems it totally orders the decision-maker's set of actions. We discuss applications to games, insurance, investment hedging, and security design.
\end{abstract}
\keywords{Decision-making under Uncertainty; Stochastic Dominance; Risk Aversion; Robust Comparisons}\\
\jel{D0; D8} 

\section{Introduction}

\textit{Risk} is a central component of decision making under uncertainty. By now it is understood that in many situations, decision-makers (DMs) dislike risk, and such aversion affects behavior. Moreover, there is a standard notion of what it means for one DM to be more risk averse than another: a more risk-averse agent eschews any risk that a less-risk averse one does. In the expected-utility realm, the more risk-averse DM has a utility that is ``more concave'' than the other's.

A closely-related literature then asks how to compare \emph{lotteries} in a way that speaks to such utility transformations. One canonical question is: when does every expected-utility DM with a concave utility rank one lottery above another? A second question fixes a benchmark and asks when a preference comparison is preserved under \emph{all} increases in risk aversion. These questions lead to different partial orders on distributions; second-order stochastic dominance, and (for various versions of the preservation problem) single-crossing conditions on the lotteries' distributions.

In a number of economic settings, however, the primitive object is not a lottery but an \emph{action} (or strategy)--a state-contingent payoff profile. A belief (a probability distribution over states) turns an action into a lottery, and the same action can induce very different lotteries under different beliefs. As a result, lottery orders do not answer the natural question we face in many applications: when can we rank actions in a way that does not depend on fixing a single belief, and whose defining feature is robustness to increases in risk aversion?

This distinction matters in exactly the environments where beliefs are not a fixed, single-valued, and \textit{known} background object. For example, in models with learning, different pieces of information generate different beliefs, and we would like the comparison of actions to hold across these events. In games, beliefs are conjectures about opponents' strategies and are part of the equilibrium description, so the relevant randomness is endogenous to the strategic environment. Furthermore, as analysts we often want conclusions that do not hinge on a single specified belief, either because beliefs are heterogeneous across a population or because we want robust predictions that do not hinge on what we assume agents think. These considerations motivate an ordering defined on actions that is robust to variation in beliefs, indeed, one that is fundamentally \textit{distribution free}.

The main contribution of this paper is to provide such an action-level theory. We define a robust ``safer-than'' relation on actions that requires an agent's preference between two actions to survive every strictly increasing, concave transformation of her utility function and to hold regardless of beliefs. We then show that this behavioral requirement admits a sharp, tractable characterization, equivalent to a collection of simple state-by-state inequalities, to a single-crossing property of the lotteries induced by the actions under every belief, and to a collection of rankings according to second-order stochastic dominance.

To elaborate, we say that one action \(a\) is \emph{safer} than another action \(b\), \(a \succeq_S b\), if the set of beliefs at which action \(a\) is preferred to \(b\) grows--in the set inclusion sense--as the DM becomes more risk averse. That is, safer actions become more attractive whenever the DM's utility function is transformed by a strictly increasing, concave function. Our main result, Theorem \ref{thm:central}, provides a full characterization of this relation. Our principal equivalence is that \(a\) is safer than \(b\) if and only if for any pair of states in which \(a\) is strictly preferred to \(b\) in one and \(b\) to \(a\) in the other, the payoffs to \(a\) lie within the convex hull of the payoffs to \(b\). Thus, safety is an \textit{ordinal} notion. 

Another pair of equivalences we prove in the theorem connects to the classic literature. For any fixed belief, an action induces a lottery over monetary prizes. We establish that \(a\) is safer than \(b\) if and only if for any belief, the cdfs of the induced lotteries cross once. This, in turn, allows us to equate safety with a collection of rankings of the lotteries according to second-order stochastic dominance. Thus, \textit{we provide the distribution-free analog of both the single-crossing relation and the second-order stochastic dominance relation between lotteries, and show, moreover, that they coincide}.

The safety relation is not transitive, which guides us to monotone decision problems, in which the sets of states and actions are totally ordered, higher states are preferred by the agent, and higher actions are better in higher states. Proposition \ref{bigequivalence} reveals that in monotone environments, safety is not only transitive, but totally orders the set of actions. Notably, in this class of problems, safer actions are lower actions, revealing that the optimality of higher actions is more sensitive to the agent's risk aversion.

\medskip

\noindent \textbf{\textcolor{BlueViolet}{Applications:}} We discuss the usefulness of the safer-than relation in four environments. First, we study games, using safety to present comparative statics for a broad class of binary ``contribution games'' that include as special cases voting and whistleblowing. We also connect safety to risk-dominance and the global games literature. Second, we use our relation to rank securities according to their sensitivity to an investor's risk aversion. We show that safety coincides with single-crossing (``steepness'') of the securities, which is a well-established measure for the information-sensitivity of the securities. Third, we use safety to rank varieties of insurance.

Fourth, we use our relation to evaluate robust hedging in an investment setting. Given an investor's current holdings, we formulate a binary relation between assets: one asset ``hedges'' better than another if the set of beliefs justifying it expands as the investor becomes more risk averse. Importantly, the belief-free formulation of our relation enables us to allow for arbitrary correlation between the investor's current holdings and the assets under scrutiny. Finally, in a companion paper, \cite{peasewhitmeyerauctions}, we show that bids in broad classes of auctions can be ranked according to safety, before using this observation to generalize and elucidate classic comparative statics results on bidders' risk aversion in auctions.

\subsection{Related Work}

\cite*{yaari69} introduces an important comparative notion of risk aversion: ``Mr. A is more risk averse than Mr. B if...every gamble which is acceptable to A is acceptable to B."\footnote{Similarly, \cite*{ghirardato2002ambiguity} define when preferences are ``more ambiguity averse" than others.}
We reveal that there is a natural way to broaden this notion beyond a total absence of risk. That is, \citeauthor{yaari69}'s behavioral characterization of ``less risk averse'' can be redefined to mean increased willingness to pick an action over a safer one, instead of just the safest one.

\cite*{hammond1974simplifying}, \cite*{lambert1979attitudes}, \cite*{karlin1963generalized}, and \cite*{jewitt1987risk} all contain results concerning when an action that is preferred to another given some utility function must still be preferred following any increase in the agent's risk aversion. The key difference is that the condition in these papers is on the \textit{distributions} over wealth obtained by the agent as a result of her choice of actions. In this context, one contribution of our paper is to formulate a way of comparing actions' comparative robustness that is \textit{distribution-free}. In our main theorem, we connect our concept of safety to the single-crossing property identified by these works and reveal that safety is equivalent to a robust notion of single-crossing of the induced distributions over wealth for all subjective beliefs.

\cite*{flexibilitypaper} is also related to our work. That paper studies transformations of decision problems that render information more valuable to a DM. Here, we study a particular variety of transformation--an increase in the DM's risk aversion--and focus on its effect on the optimality of various actions.

This paper also harkens to the comparative statics literature; see, e.g., \cite*{milgrom1994monotone}, \cite*{edlin1998strict}, and \cite*{athey2002monotone}. We also vary a parameter, the DM's risk aversion, and ask how this affects the DM's behavior. However, we focus on comparisons between actions and make our relation quite demanding: the enlargement of the set of beliefs at which an action is preferred to another must arise for any monotone concave transformation of the DM's utility. In that light, our paper relates to the works that study the aggregation of the single-crossing property (\cite*{quah2012aggregating}, \cite*{choi2017ordinal}, and \cite*{kartik2023single}). We elaborate on the connection of our main result to \cite*{quah2012aggregating} in \S\ref{sec:srm}.

Our paper also adds to a literature that aims to measure the riskiness or ambiguity of lotteries or acts. \cite*{ROTHSCHILD1970225} introduces a seminal definition for comparing the riskiness of lotteries with the same mean. \cite*{aumann2008economic} propose a different measure of lottery riskiness that uses the reciprocal of the absolute risk aversion of a DM, while \cite*{foster2009operational} extend this definition to depend only on the gamble, not on the DM, defining a wealth level below which it is risky to accept the lottery. Our measure of safety, however, is not a ranking of lotteries; rather it is a belief-free comparison that relies only on the state-contingent payoffs. \cite*{jewitt2017ordering} rank the ambiguity of (Anscombe-Aumann) acts. Their findings are closer to the spirit of these other works, rather than ours: one act is more ambiguous than another if it is (in a sense) an ambiguity-preserving spread of the other.

Finally, our work is related to the striking results of \cite*{battigalli2016note} and \cite*{weinstein2016effect}, who prove that increased risk aversion on the part of the DM enlarges the set of justifiable actions, the actions that are optimal at some belief. That is, a justifiable action remains justifiable if the DM becomes more risk averse; and increased risk aversion may render actions optimal that had previously been strictly dominated. We, instead, study the properties of decision problems and actions therein for which increased risk aversion enlarges the set of beliefs at which some actions are optimal. Informally, as risk aversion increases, more actions become justifiable, \textit{but some (safer ones) become more justifiable than others.}

\section{Model}

There is an unknown state of the world \(\theta\), which is an element of some topological space of states \(\Theta\), endowed with the Borel \(\sigma\)-algebra. We specify further that \(\Theta\) is compact and metrizable. We denote the set of all Borel probability measures on \(\Theta\) by \(\Delta \equiv \Delta \left(\Theta\right)\). Our protagonist is a decision-maker (DM) with a set of actions \(A \subseteq \mathbb{R}^{\Theta}\), where each action is a continuous function from the set of states, \(\Theta\), to the set of outcomes (monetary payoffs) \(\mathbb{R}\). The decision-maker has a utility function defined on the outcome space \(u \colon \mathbb{R} \to \mathbb{R}\), which we assume is continuous and strictly increasing.

Let \(a_\theta \in \mathbb{R}\) denote the payoff to action \(a\) in state \(\theta\). The DM has a subjective belief \(\mu \in \Delta\); and she is a subjective expected-utility (EU) maximizer. We assume that no action in \(A\) weakly dominates another: for all \(a,b \in A\), there exists some \(\theta, \theta' \in \Theta\) such that \(a_\theta > b_\theta\) and \(b_{\theta'} > a_{\theta'}\). We say that the DM becomes more risk averse if her utility function is \(\hat{u}\) where \(\hat{u} = \phi \circ u\) for some strictly increasing concave \(\phi\).

For any two actions \(a, b \in A\) (\(a \neq b\)), we define the set \(P_{a, b}\left(a\right)\) to be the subset of the probability simplex on which action \(a\) is weakly preferred to \(b\) under \(u\); formally,
\[P_{a,b}\left(a\right) \coloneqq \left\{\mu \in \Delta \colon \mathbb{E}_{\mu}u\left(a_\theta\right) \geq \mathbb{E}_{\mu}u\left(b_\theta\right)\right\}\text{.}\]
By assumption this set is non-empty and--when \(\Delta\) is finite dimensional--of full dimension in \(\Delta\). When the utility function is \(\hat{u}\), we define the set \(\hat{P}_{a,b}\left(a\right)\) in the analogous manner.

\begin{definition}
    Action \(a\) is \emph{safer} than action \(b\) if for any strictly increasing \(u\) and any strictly increasing, concave \(\phi\), \(P_{a,b}\left(a\right) \subseteq \hat{P}_{a,b}\left(a\right)\); i.e., the set of beliefs at which action \(a\) is preferred to \(b\) increases in size as the DM becomes more risk averse.
\end{definition}
Equivalently, action \(a\) is safer than action \(b\) if for every strictly increasing and concave function \(\phi\), strictly increasing function \(u\), and belief \(\mu \in \Delta\), if \(\mathbb{E}_{\mu}u\left(a_\theta\right) \geq \mathbb{E}_{\mu}u\left(b_\theta\right)\), then \(\mathbb{E}_{\mu}\phi \circ u\left(a_\theta\right) \geq \mathbb{E}_{\mu}\phi \circ u\left(b_\theta\right)\). Let \(a \succeq_S b\) denote the binary relation \textit{``action \(a\) is safer than action \(b\).''} The strict relation, \(a \succ_S b\) denotes \(a \succeq_{S} b\) but \(b \not\succeq_{S} a\).

\subsection{Single-Crossing}\label{sec:singlecross}

A classical question concerns which lotteries are made comparatively more attractive as agents become more risk averse. This is similar in spirit to our safety notion but there is an important difference: the classical relation is between lotteries, which are random variables, whereas ours is between vectors of (state-dependent) payoffs.

Nevertheless, there is an intimate connection between these two concepts. Given belief \(\mu \in \Delta\), let \(X^\mu_a\) and \(X^\mu_b\) denote the real-valued random variables induced by the actions \(a\) and \(b\), with (respective) cumulative distribution functions (cdfs)
\[
F_a^{\mu}(v) \coloneqq \mu\left(\left\{\left.\theta \right| a_\theta \leq v\right\}\right),
\quad \text{and} \quad
F_b^{\mu}(v) \coloneqq \mu\left(\left\{\left.\theta \right| b_\theta \leq v\right\}\right).
\]
We introduce a familiar definition:
\begin{definition}\label{singlecrossdef}
    \(F^{\mu}_{a}\) \emph{single-crosses} \(F^{\mu}_{b}\) from below, \(F^{\mu}_{a} \succeq_{sc} F^{\mu}_{b}\), if there exists a \(\bar{v} \in \mathbb{R}\) such that for all \(v \in \mathbb{R}\), \(v < \bar{v}\) implies \(F^{\mu}_{b}(v) \geq F^{\mu}_{a}(v)\) and \(v \geq \bar{v}\) implies \(F^{\mu}_{a}(v) \geq F^{\mu}_{b}(v)\).
\end{definition}

Then, \cite*{hammond1974simplifying}, \cite*{lambert1979attitudes}, \cite*{karlin1963generalized}, and \cite*{jewitt1987risk} all contain the following result that highlights the importance of single-crossing in understanding risk aversion. From \cite*{jewitt1987risk},
\begin{theorem}\label{single-crossingthm}
    Let \(u\) be strictly increasing and \(F^{\mu}_{a} \succeq_{sc} F^{\mu}_{b}\). Then, \[\int u(v) dF^{\mu}_{a}\left(v\right) \geq \int u(v) dF^{\mu}_{b}\left(v\right) \quad \Longrightarrow \quad \int \phi \circ u \left(v\right) dF^{\mu}_{a}\left(v\right) \geq \int \phi \circ u\left(v\right)  dF^{\mu}_{b}\left(v\right),\]
    whenever \(\phi\) is strictly increasing and concave.
\end{theorem}
To reiterate, single-crossing is a distribution-specific property. It is a statement about cdfs--in the economic context, a characteristic of lotteries. Accordingly, the classical results are comparing specific distributions. In contrast, our concept of safety is prior-free--it is a statement merely about (state-dependent) payoffs. Our main result, Theorem \ref{thm:central}, will connect the two.

\section{Main Result}

We define \(\mathcal{A}\) to be the set of states in which \(a\) is strictly preferred to \(b\), \(\mathcal{B}\) to be the set of states in which \(b\) is strictly preferred to \(a\), and \(\mathcal{C}\) to be the set of states in which the DM is indifferent between \(a\) and \(b\): \[\mathcal{A} \coloneqq \left\{\theta \in \Theta \colon a_{\theta} \ > \ b_{\theta}\right\}, \quad \mathcal{B} \coloneqq \left\{\theta \in \Theta \colon a_{\theta} \ < \ b_{\theta}\right\}, \quad \mathcal{C} \coloneqq \left\{\theta \in \Theta \colon a_{\theta} \ = \ b_{\theta}\right\}\text{.}\]
\begin{theorem}\label{thm:central}
The following are equivalent:
\begin{enumerate}
    \item\label{bullet1} \(a \succeq_S b\).
    \item\label{bullet2} For any pair \((\theta, \theta') \in \mathcal{A} \times \mathcal{B}\), \(b_{\theta'} \geq a_\theta\) and \(a_{\theta'} \geq b_\theta\).
    \item\label{bullet3} For all \(\mu \in \Delta\), \(F_a^{\mu} \succeq_{sc} F_b^{\mu}\).
    \item\label{bullet4} For all \(\mu \in \Delta\), if \(\mathbb{E}_{\mu} a_\theta \geq \mathbb{E}_{\mu} b_\theta\), then there exists a lottery \(\widetilde{X}\) such that \(X^\mu_a\) first-order stochastically dominates (FOSD) \(\widetilde{X}\) and \(\widetilde{X}\) second-order stochastically dominates (SOSD) \(X^\mu_b\). Moreover, if \(\mathbb{E}_{\mu} a_\theta = \mathbb{E}_{\mu} b_\theta\), then \(X^\mu_a\) SOSD \(X^\mu_b\).
\end{enumerate}
\end{theorem}
\begin{proof}
    Several of these implications follow easily from known results. In particular, Theorems 4.A.6 and 4.A.22 in \citet{shaked2007stochastic} yield \ref{bullet3} \(\Longrightarrow\) \ref{bullet4}; with the equal-mean component of \ref{bullet4} following from Theorem 3.A.44 in \citet{shaked2007stochastic}. \ref{bullet4} \(\Longrightarrow\) \ref{bullet2} due to the equivalence of nested supports and second-order stochastic dominance for binary random variables, and Theorem \ref{single-crossingthm} gives \ref{bullet3} \(\Longrightarrow\) \ref{bullet1}. Accordingly, it remains to show \ref{bullet1} \(\Longleftrightarrow\) \ref{bullet2} \(\Longrightarrow\) \ref{bullet3}. We prove these lemmas, completing the proof of the theorem, in Appendix \ref{thm:centralproof}.
\end{proof}

\begin{figure}
    \centering
    \includegraphics[width=1\linewidth]{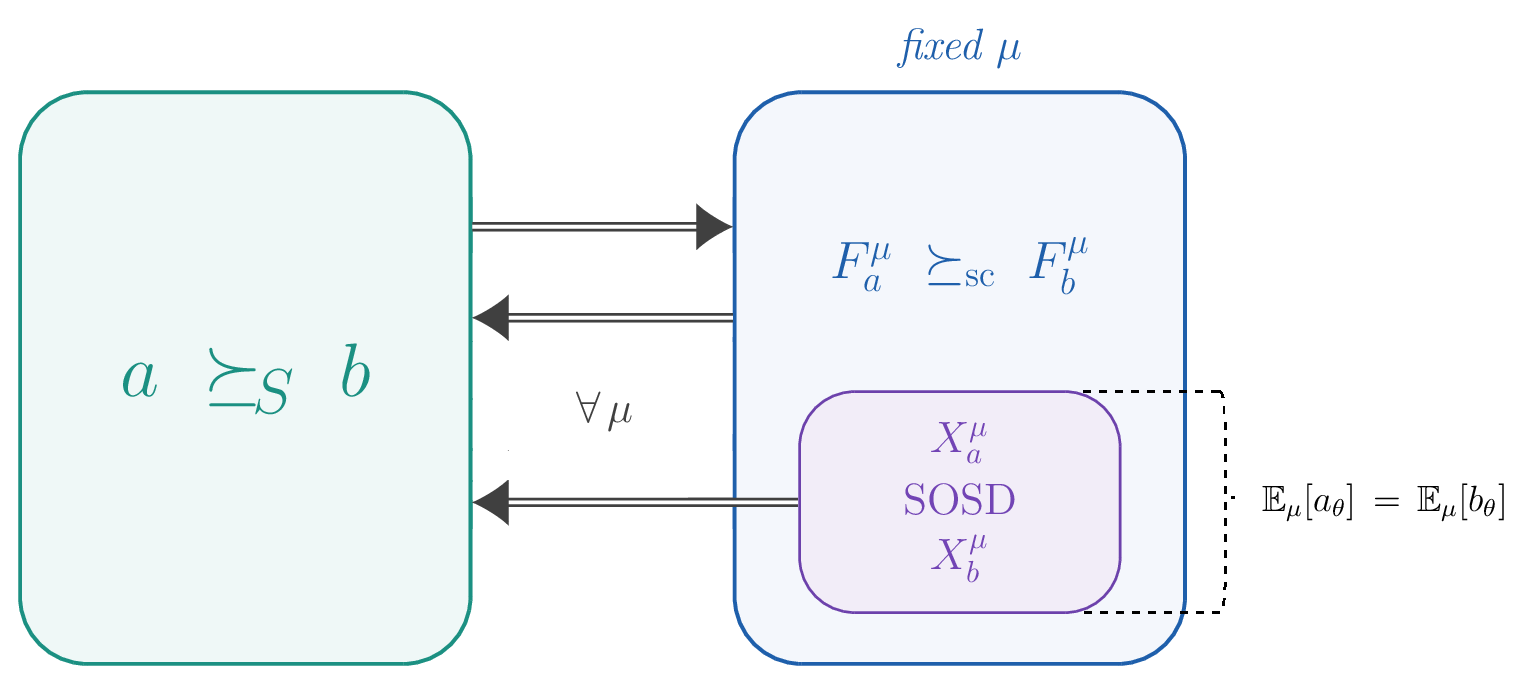}
    \caption{Theorem \ref{thm:central}. Connecting safety to SOSD and single-crossing.}
    \label{fig:thm31}
\end{figure}

Let us consider each piece of Theorem \ref{thm:central}. The equivalence Property \ref{bullet1} \(\Longleftrightarrow\) Property \ref{bullet2} states that $a$ is safer than $b$, if and only if the payoffs of $a$ lie within the convex hull of the payoffs of $b$. This is our principal characterization and we discuss its intuition in greater detail in \S\ref{ss:intuition} below. Properties \ref{bullet3} and \ref{bullet4} link the notion of safety to the classical ways of discussing risk with specific belief distributions, namely single-crossing and stochastic dominance. 

The implication Property \ref{bullet1} \(\Longrightarrow\) Property \ref{bullet3} says that if $a$ is safer than $b$, then a robust version of single-crossing holds: it must be satisfied \emph{for all} $\mu\in \Delta$. Its converse is exactly the classical results. Similarly, Property \ref{bullet4} connects safety to first- and second-order stochastic dominance \emph{for every} $\mu$ such that the expected monetary payoff from $a$ is higher than $b$. We visualize the connections of the items to each other in Figure \ref{fig:thm31}.

It is important to note that we have chosen one particular definition for ``safety.'' Namely, that the set of beliefs at which $a$ is preferred grows in risk aversion. While this is not the only possible definition for safety, we believe that it is the best way to robustly rank the riskiness of actions, given its strong connection to both single-crossing and stochastic dominance. When we discuss games in \S\ref{ss:games}, we connect our safety relation to risk dominance.

\subsection{Intuition}\label{ss:intuition}

Toward gaining intuition for the equivalence of \ref{bullet1} and \ref{bullet2} in Theorem \ref{thm:central}--the equivalence of \(a \succeq_S b\) and \(b_{\theta'} \geq a_\theta\) and \(a_{\theta'} \geq b_\theta\) for every pair of states \((\theta, \theta') \in \mathcal{A} \times \mathcal{B}\)--we briefly specialize to the binary state environment.

When there are two states, \(0\) and \(1\), for two actions \(a, b \in A\), denote \(\alpha_\theta \equiv u\left(a_\theta\right)\) and \(\beta_\theta \equiv u\left(b_\theta\right)\). As neither action dominates the other, and as we could just relabel the actions, we specify without loss of generality that \(a_0 > b_0\) and \(b_1 > a_1\) so that if the DM knows that the state is \(0\), she strictly prefers \(a\) and if she knows that the state is \(1\), she strictly prefers \(b\). 

Consider the DM's choice between actions \(a\) and \(b\). Let \(\mu \in \left[0,1\right]\) be the probability that the state is \(1\), so that she prefers action \(a\) if and only if
\[\left(1-\mu\right)\alpha_0 + \mu \alpha_1 \geq \left(1-\mu\right) \beta_0 + \mu\beta_1 \quad \Longleftrightarrow \quad \mu \leq \frac{\alpha_0 - \beta_0}{\alpha_0 - \beta_0 + \beta_1 - \alpha_1} \equiv \bar{\mu}\text{.}\]
Analogously, under the transformed utility, the indifference belief is
\[\hat{\mu} \equiv \frac{\phi\left(\alpha_0\right) - \phi\left(\beta_0\right)}{\phi\left(\alpha_0\right) - \phi\left(\beta_0\right) + \phi\left(\beta_1\right) - \phi\left(\alpha_1\right)}\text{.}\]
By definition, action \(a\) is safer than action \(b\) if it is chosen for more beliefs after the change in utility; that is, if \(\bar{\mu} \leq \hat{\mu}\). This translates to 
\[\tag{\(1\)}\label{eq:safer}
\frac{\phi(\beta_1) - \phi(\alpha_1)}{\beta_1 - \alpha_1} \leq \frac{\phi(\alpha_0) - \phi(\beta_0)}{\alpha_0 - \beta_0}\text{.}\]

What is the meaning of this condition? The right-hand side is the secant line to the concave \(\phi\) at points \(\alpha_0\) and \(\beta_0\); that is, the average slope of \(\phi\) between \(\beta_0\) and \(\alpha_0\). In other words, it is the marginal benefit from ``correctly" choosing \(a\) in state \(0\) (because \(a\) is better if the state is \(0\)) under the transformation \(\phi\). Therefore, if action \(a\) is safer than action \(b\), its marginal benefit in its ``correct" state is higher under the transformation, and it will be chosen more. 

Using Inequality \eqref{eq:safer}, we can characterize the ``safer than'' binary relation, \(\succeq_S\), exactly when there are two states:
\begin{proposition}\label{twostatesaferthanequiv}
    For actions \(a,b \in A\), \(a \succeq_S b\) if and only if \(b_1 \geq a_0\) and \(a_1 \geq b_0\).
\end{proposition}
\begin{proof}
\(\left(\Leftarrow\right)\) By the strict monotonicity of \(u\), \(b_1 \geq a_0\) and \(a_1 \geq b_0\) if and only if \(\beta_1 \geq \alpha_0\) and \(\alpha_1 \geq \beta_0\). Then, the Three-chord lemma (Theorem 1.16 in \cite*{phelps2009convex}) yields
    \[\label{ineq1}\tag{\(2\)}\frac{\phi\left(\alpha_0\right) - \phi\left(\beta_0\right)}{\alpha_0 - \beta_0} \geq \frac{\phi\left(\beta_1\right) - \phi\left(\beta_0\right)}{\beta_1 - \beta_0} \text{.}\]
    
    Likewise, \(\beta_1 > \alpha_1 \geq \beta_0\) plus the Three-chord lemma imply
    \[\label{ineq2}\tag{\(3\)}\frac{\phi\left(\beta_1\right) - \phi\left(\beta_0\right)}{\beta_1 - \beta_0} \geq \frac{\phi\left(\beta_1\right) - \phi\left(\alpha_1\right)}{\beta_1 - \alpha_1} \text{.}\]
    Combining Inequalities \ref{ineq1} and \ref{ineq2} yields Inequality \ref{eq:safer}.

\medskip

    \noindent \(\left(\Rightarrow\right)\) See Appendix \ref{proofsafe}.\end{proof}
Recalling that \(\alpha_0 > \beta_0\) and \(\beta_1 > \alpha_1\), we can rephrase this proposition as saying that \(a\) is safer than \(b\) if and only if the payoffs of \(a\) lie in the convex hull of the payoffs of \(b\). In other words, \(a \succeq_S b\) if and only if choosing \(a\) and being ``wrong" (because the state is actually \(1\)) is (weakly) not worse than choosing \(b\) and being ``wrong," while choosing \(a\) and being ``right" is not better than choosing \(b\) and being ``right." 

With this in mind, we can sketch an alternative, elementary, proof of the proposition: in the two-state setting, fixing \(u\), the lottery in utils induced by \(a\) and belief \(\mu\) at the indifference belief (given the utility function) between \(a\) and \(b\) is a mean-preserving contraction of that corresponding to \(b\) if and only if the payoffs to \(a\) lie in the convex hull of those according to \(b\). Accordingly, at this original indifference belief, the DM must continue to prefer \(a\) to \(b\) after any increase in risk aversion if and only if this convex-hull property is satisfied.

To extend Proposition \ref{twostatesaferthanequiv}'s binary-state finding to a general state space requires only the observation that it suffices to compare action by action and state by state along particular edges of \(\Delta\).\footnote{To elaborate, the set of beliefs on which $a$ is preferred to $b$ is the intersection of \(\Delta\) and a half-space (\(\mathbb{E}_\mu u(a_\theta) \geq \mathbb{E}_\mu u(b_\theta)\)). This is a convex set. Moreover, its extreme points are distributions that either put probability \(1\) on states in \(\mathcal{A}\) or make the DM indifferent between $a$ and $b$. The Krein-Milman theorem then allows us to simplify the setting to a collection of binary-state environments, as characterized in Proposition \ref{twostatesaferthanequiv}.} Specifically, consider states \(\theta \in \mathcal{A}\) and \(\theta' \in \mathcal{B}\). Comparing payoffs for only these two states, our result from Proposition \ref{twostatesaferthanequiv} applies directly so that, if the probability of all other states is 0, \(a\) is safer than \(b\) if and only if $a_{\theta} \leq b_{\theta'}$ and $b_{\theta} \leq a_{\theta'}$, as before. Doing this for every pair of states in which \(\theta \in \mathcal{A}\) and \(\theta' \in \mathcal{B}\) yields the result.

\subsection{Signed-Ratio Monotonicity}\label{sec:srm}
At a high level the safety relation is intimately connected to the signed-ratio monotonicity condition of \cite{quah2012aggregating}. Given a partially-ordered set \(\left(S,\geq\right)\) and a set of states of the world \(\Theta\), for distinct \(\theta, \theta' \in \Theta\), functions \(f_{\theta}(s)\) and \(f_{\theta'}(s)\), satisfy signed-ratio monotonicity if \(f_{\theta}\) and \(f_{\theta'}\) satisfy the single-crossing property in \(s\) and at any \(s'\) such that \(f_{\theta}(s^\prime) < 0 < f_{\theta'}(s^\prime)\), \[\label{srmono} \tag{\(4\)}s'' > s' \quad \Rightarrow \quad -\frac{f_{\theta}(s')}{f_{\theta'}(s')} \geq -\frac{f_{\theta}(s'')}{f_{\theta'}(s'')}\text{.}\] Signed-ratio monotonicity is an important property: as \cite{quah2012aggregating} show in their main theorem, \(\int f_{\theta}(s)d\mu(\theta)\) is single-crossing in \(s\) for any \(\mu\) if and only if \(f_{\theta}(s)\) and \(f_{\theta'}(s)\) satisfy signed-ratio monotonicity for all \(\theta, \theta' \in \Theta\).

What is the connection to safety? One can let \(\left(S,\geq\right)\) index risk aversion. Then, for \(a, b \in A\), defining \[f_{\tilde{\theta}}(s'') = u(b_{\tilde{\theta}}) - u(a_{\tilde{\theta}}), \quad \text{and} \quad f_{\tilde{\theta}}(s') = \phi \circ u(b_{\tilde{\theta}}) - \phi \circ u(a_{\tilde{\theta}}) \text{ for } \tilde{\theta} \in \left\{\theta,\theta'\right\}\text{,}\] and plugging them into the right-hand inequality in Expression \ref{srmono}, we obtain Inequality \ref{eq:safer}, which, as we show, holds if and only if the payoffs to \(a\) lie in the convex hull of \(b\)'s payoffs. In fact, signed-ratio monotonicity is simply \textit{a generalized notion of ``increased concavity,''} as shown in \cite{curelloetal}.

\section{Further Results}

\subsection{Beyond Expected Utility}\label{sec:beyond}

This connection to single crossing allows us to briefly examine safety outside of the expected-utility (EU) environment. We suppose that for any belief \(\mu \in \Delta\), the DM's preferences over actions are represented by a binary relation, \(\succeq\), that is a ``risk preference'' in the terminology of \cite*{maccheroniinsurancerisk} (i.e., is transitive and law invariant). We refer readers to this paper for further details: in particular, \cite{maccheroniinsurancerisk} note that this class includes \textit{all} probabilistically-sophisticated preferences. We further impose that the DM's preferences respect first-order stochastic dominance (FOSD).

We adapt their terminology in the following definition. Given action \(a\), recall that \(X^{\mu}_{a}\) denotes the lottery induced by belief \(\mu\). Then,
\begin{definition}
    A risk preference, \(\succeq\), is \textit{strongly risk averse} if for all \(a, b \in A\), if \(\mu\) is such that \(X^{\mu}_{a}\) second-order stochastically dominates (SOSD) \(X^{\mu}_{b}\), \(X^{\mu}_{a} \succeq X^{\mu}_{b}\). \(\succeq\) is \textit{risk neutral} if it admits an expected-value representation:
    \[X^{\mu}_{a} \succeq X^{\mu}_{b} \quad \Longleftrightarrow \quad \mathbb{E}_\mu a_\theta \geq \mathbb{E}_\mu b_\theta\text{.}\]
\end{definition}

Safer actions are precisely those that cannot become less enticing to a strongly risk-averse DM than a risk-neutral DM:
\begin{corollary}
    Action \(a\) is safer than action \(b\) if and only if a strongly risk-averse DM prefers \(a\) to \(b\) whenever a risk-neutral DM does.
\end{corollary}

\subsection{Monotone Decision Problems}\label{s:monotone}
We now specialize to a class of decision problems that are monotone in both states and actions. In these problems, larger states are preferred by the agent, and larger actions are (vaguely) better in larger states (single-crossing differences). Monotone problems arise in a number of contexts in economics including production and hiring by firms, investing (portfolio choice), and consumption-savings decisions. We find that not only does our concept of safety apply easily in these problems, but that the extra structure of monotonicity yields many desirable properties. 

We call a decision problem \(\left(A,\Theta,u\right)\) \textit{monotone} if \(A\) is endowed with a total order \(\trianglerighteq\), \(a_\theta\) is increasing in \(\theta\) for all \(a \in A\), and the \emph{single-crossing differences} property holds:
    \[a' \triangleright a \text{ and } a_{\theta}' \underset{(>)}{\geq} a_\theta \ \Rightarrow \ a_{\theta'}' \underset{(>)}{\geq} a_{\theta'} \text{ for any } \theta' \geq \theta\text{.}\]
Theorem \ref{thm:central} delivers the following equivalence result.
\begin{proposition}\label{bigequivalence}
    For monotone decision problems, \(a\) is safer than \(b\) if and only if \(b \trianglerighteq a\). Moreover, \(\succeq_S\) is transitive on monotone decision problems.
\end{proposition}
\begin{proof}
    \(a \succeq_S b\) if and only if for all \(\left(\theta,\theta'\right) \in \mathcal{A} \times \mathcal{B}\), \(b_{\theta'} \geq a_\theta > b_\theta\) and
    \(b_{\theta'} > a_{\theta'} \geq b_\theta\), which holds in monotone decision problems if and only if \(b \trianglerighteq a\). As \(\trianglerighteq\) is transitive, so too is \(\succeq_S\).
\end{proof}

    In fact, as \(\trianglerighteq\) totally orders \(A\), Proposition \ref{bigequivalence} implies\footnote{In Appendix \ref{noordersec}, we show that if we impose less structure on the state space, we can, nevertheless, produce refinements of safety that are transitive. Essentially, we generate the structure needed for transitivity via the refinements themselves.}
    \begin{corollary}
        For any monotone decision problem, \(\succeq_{S}\) totally orders the set of actions.
    \end{corollary}

\subsection{Learning and Incomplete Preferences}\label{learning}

Our relation allows for robust predictions when the DM is not assumed to have a single belief. Two situations in which this manifests are when the DM obtains information before taking an action, or when the DM has the incomplete preferences of \cite{bewley2002knightian}.

We explore information acquisition first and make two assumptions about the DM's behavior. First, we maintain our standing rationality assumption. Second, let \(k \colon \Delta \to [0,1]\) denote the DM's probability of choosing action \(a\) under utility \(u\) and \(\hat{k}\) her analogous choice rule under utility \(\hat{u}\). Given a binary set of actions \(A = \left\{a,b\right\}\), we say a DM's behavior is \textit{consistent} if for any two utility functions \(u\) and \(\hat{u}\) and for all \(\mu \in \Delta\), whenever \(\mathbb{E}_\mu u(a_\theta) = \mathbb{E}_\mu u(b_\theta)\) and \(\mathbb{E}_\mu \hat u(a_\theta) = \mathbb{E}_\mu \hat u(b_\theta)\), then the probability the DM selects action \(a\) at belief \(\mu\) under utility \(u\) equals the probability the DM selects action \(a\) at belief \(\mu\) under utility \(\hat{u}\). That is, consistency means that, holding fixed the DM's belief, her utility function does not alter how she resolves indifference between actions. Thus, consistency means that for all \(\mu \in \Delta\) such that \(\mathbb{E}_\mu u(a_\theta) = \mathbb{E}_\mu u(b_\theta)\) and \(\mathbb{E}_\mu \hat u(a_\theta) = \mathbb{E}_\mu \hat u(b_\theta)\), \(k(\mu) = \hat k(\mu)\).

Given a full-support prior \(\mu_0 \in \Delta\), information acquisition corresponds to a Borel probability distribution \(\Phi \in \Delta \Delta (\Theta)\) that is Bayes' plausible: \(\mathbb{E}_\Phi \mu = \mu_0\). Then, 
\begin{corollary}\label{c:learn}
    Fix a full-support prior \(\mu_0 \in \Delta\). Then, \(a \succeq_S b\) if and only if for any Bayes-plausible \(\Phi\), any utilities \(u\) and \(\hat{u}\) with \(\hat{u}\) more risk averse than \(u\), and any consistent, rational \(k\) and \(\hat{k}\), the more risk-averse DM chooses action \(a\) more frequently than \(b\):
    \(\mathbb{E}_\Phi k(\mu) \leq \mathbb{E}_\Phi \hat{k}(\mu)\).
\end{corollary}

This result highlights a unique feature of the concept of safety. Because it is distribution free, it can give predictions that are robust to \emph{any} learning process. Corollary \ref{c:learn} shows that safety guarantees that a more risk-averse DM chooses action $a$ more frequently in expectation, independent of signal structure or the precision of information. 

Next, our safety relation is intrinsically compatible with \cite{bewley2002knightian}'s model of incomplete preferences (see also \cite{ghirardato2003subjective}). In that model, a DM entertains a compact and convex set of probabilities \(M \subseteq \Delta\) and prefers action \(a\) to \(b\) if and only if \(\mathbb{E}_\mu u(a_\theta) \geq \mathbb{E}_\mu u(b_\theta)\) for all \(\mu \in M\). As the safety relation is itself defined by a set-inclusion ranking of sets of beliefs, Theorem \ref{thm:central} implies
\begin{corollary}
    If \(a\) is safer than \(b\), then if a Bewley DM prefers \(a\) to \(b\), she prefers \(a\) to \(b\) for any increase in risk aversion.
\end{corollary}

\subsection{Ambiguity}\label{smooth}

Our main result extends in a natural way to the smooth ambiguity model of \cite*{klibanoff2005smooth}. Suppose our DM prefers action \(a\) to action \(b\) if and only if 
\[\mathbb{E}_{\nu} \psi\left(\int u(a_\theta)d\pi\left(\theta\right)\right) \geq \mathbb{E}_{\nu} \psi\left(\int u(b_\theta)d\pi\left(\theta\right)\right)\text{,}\]
where \(\psi\) is a monotone concave function and \(\nu \in \Delta \Pi\) is a distribution over feasible probability measures \(\pi \in \Pi \subseteq \Delta \Theta\). We specify that \(\Pi\) is compact. 

Following \cite{klibanoff2005smooth}, our DM becomes more ambiguity averse if the internal von Neumann–Morgenstern utility \(u\) stays unchanged and \(\psi\) transforms to \(\phi \circ \psi\), for some monotone concave function \(\phi\).

Understanding \(\psi\left(\int u(a_\theta)d\pi\left(\theta\right)\right)\) as a concave functional \(\psi \colon A \times \Delta \Theta \to \mathbb{R}_+\), we define \(\mathcal{A} \subset \Delta \Theta\) to be the set of priors at which \(a\) is uniquely optimal and \(\mathcal{B} = \Delta \Theta \setminus \mathcal{A}\) to be the set in which \(b\) is uniquely optimal. We further define \[\alpha_{\pi} \coloneqq \psi\left(\int u(a_\theta)d\pi\left(\theta\right)\right) \quad \text{and} \quad \beta_{\pi} \coloneqq \psi\left(\int u(b_\theta)d\pi\left(\theta\right)\right)\text{;}\]
then, applying Theorem \ref{thm:central}, obtain
\begin{proposition}
Action \(a\) is safer than action \(b\) if and only if for each \(\pi \in \mathcal{A}\) and \(\pi' \in \mathcal{B}\), \(\beta_{\pi'} \geq \alpha_{\pi}\) and \(\alpha_{\pi'} \geq \beta_{\pi}\).
\end{proposition}

Comparative smooth ambiguity aversion corresponds to a concave transformation of an object, just like comparative risk aversion. Not so for comparative ambiguity aversion for a maxmin DM. There, a DM becomes more ambiguity averse if her utility function over actions remains unchanged, but the set of beliefs over which she evaluates her expected utility (pessimistically) grows in a set-inclusion sense \citep{borie2023maxmin}. It is easy to see that ``more-ambiguity-averse'' is incompatible with safety: by introducing new feasible beliefs over which the safe action is especially bad, a less ambiguity-averse DM may prefer a safer action than a more ambiguity-averse DM.

\section{Applications}

We now scrutinize four classic settings through the lens of safety: binary-contribution games (\S\ref{ss:games}), asset-backed securities with heterogeneous beliefs and information acquisition (\S\ref{security}), insurance (\S\ref{ss:insurance}), and robust hedging (\S\ref{hedge}).

\subsection{Games}\label{ss:games}

\subsubsection{Contribution Games}
The safety relation is useful in games in which beliefs are endogenous equilibrium objects, giving comparative statics new to the literature. We consider a contribution/participation game, with examples like voting, whistleblowing, public good provision, and revolts, in which players' payoffs depend on how many others contribute. Specifically, we examine the following abstract, anonymous game with \(n\ge2\) players. Each player chooses an action \(a_i\in\left\{0,1\right\}\). Let \(m\in\left\{0,1,\dots,n-1\right\}\) denote the number of players choosing action \(1\). For each \(a\in\left\{0,1\right\}\), let \(\nu_a\colon\left\{0,1,\dots,n-1\right\}\to\mathbb R\) be the monetary payoff from choosing \(a\) against \(m\). Players evaluate monetary payoffs according to a strictly-increasing utility function \(u\) pre-transformation and \(\hat{u} = \phi \circ u\) post-transformation (where \(\phi\) is concave and strictly increasing).

For each \(p\in\left[0,1\right]\), let \(M_p\sim\mathrm{Bin}\left(n-1,p\right)\) denote the random number of other players choosing \(1\) when each of the other \(n-1\) players chooses \(1\) independently with probability \(p\).
Define
\[
V_u(a,p)\coloneqq \mathbb E\left[u\left(\nu_a\left(M_p\right)\right)\right],
\qquad \text{and} \qquad
\Delta_u(p)\coloneqq V_u(1,p)-V_u(0,p),
\]
and likewise \(V_{\hat u}(a,p)\) and \(\Delta_{\hat u}(p)\).

Define the best-reply region for action \(1\), \(S_u\coloneqq\left\{p\in\left[0,1\right]\colon \Delta_u(p)\ge0\right\}\), and the maximal symmetric equilibrium mixing probabilities
\[
\bar{p}(u)\coloneqq \sup\left(S_u\cup\left\{0\right\}\right),
\qquad \text{and} \quad
\bar{p}(\hat u)\coloneqq \sup\left(S_{\hat u}\cup\left\{0\right\}\right).
\]

\begin{proposition}\label{prop:safety_max_symm}
If action \(1\) is safer than action \(0\), then \(S_u\subseteq S_{\hat u}\), so \(\bar{p}(\hat u)\ge \bar{p}(u)\). Moreover, \(\bar{p}(u)\) (\(\bar{p}(\hat u)\)) is the maximal symmetric Nash equilibrium mixing probability under \(u\) (\(\hat u\)).
\end{proposition}
\begin{proof}
See Appendix \ref{a:contribution}.
\end{proof}

Before we consider how this result applies to some examples, it is helpful to put it in context within the contribution games literature. Papers in this field typically characterize equilibrium participation probabilities and study comparative statics of these probabilities with respect to participation cost, number of players, or the complementarity of actions (e.g., \cite{palfrey1984participation}, \cite{palfrey1985voter}, \cite{bagnoli1989provision}). In addition, the vast majority of contribution models assume risk neutrality.\footnote{There are some exceptions. \cite{teyssier2012inequity} considers constant relative risk aversion in a sequential public goods game.} In contrast, our concept of safety allows us to ask a comparative statics question largely absent from the literature: how does equilibrium participation change as risk aversion increases? Or alternatively, how robust is equilibrium participation to risk aversion? Proposition \ref{prop:safety_max_symm} tells us that if action 1 is safer, then increased risk aversion expands the set of beliefs supporting participation.  

We apply this result to two specific games. First, take a majority voting game with a contingent commitment cost. There are \(n\) players voting on a proposal, on which a majority of players must vote for it to pass. Interpret action \(1\) as ``vote yes and commit'' and action \(0\) as ``vote no,'' and the voting threshold is \(k\coloneqq\left\lfloor n/2\right\rfloor+1\). Let \(B\) be the payoff to each player if the proposal fails and \(G>B\) the payoff if it passes. In addition, let \(c\in\left(0,G-B\right)\), where \(c\) is an implementation cost incurred if and only if the proposal passes by those who voted yes. The proposal is a public good in the sense that everyone is better off if it is passed, but all players would rather receive the benefit without having to implement the policy themselves. Thus, if \(m\) other players vote \(1\), the monetary payoff functions are
\[
\nu_1(m)=
\begin{cases}
G-c,&m\ge k-1,\\
B,&m\le k-2,
\end{cases}
\qquad \text{and} \qquad
\nu_0(m)=
\begin{cases}
G,&m\ge k,\\
B,&m\le k-1.
\end{cases}
\]
\[\mathcal{A} = \left\{m \colon m = k-1\right\}, \quad \mathcal{B} = \left\{m \colon m \geq k\right\} \quad \text{and} \quad \mathcal{C} = \left\{m \colon m \leq k-2\right\}.\]
Thus,
\begin{corollary}\label{prop:majority_contingent_safer}
Action \(1\) is safer than action \(0\). Consequently, \(\bar{p}(\hat u)\ge \bar{p}(u)\).\end{corollary}

Next, let us look at whistleblowing. Suppose that if at least one player chooses action \(1\) (``report''), an investigation occurs and all players receive benefit \(G\); if no one reports, all players receive \(B\), where \(G>B\). A reporter pays a private cost \(c\in\left(0,G-B\right)\) regardless of what others do.
Thus, if \(m\) other players report,
\[
\nu_1(m)=G-c \text{ for all }m,
\qquad
\nu_0(0)=B,\qquad \text{and} \qquad \nu_0(m)=G \quad \text{for }m\ge1.
\]
\begin{corollary}
Action \(1\) is safer than action \(0\). Consequently \(\bar{p}(\hat u)\ge \bar{p}(u)\).
\end{corollary}

In both of these examples, the set of beliefs at which the costly, but socially beneficial action is chosen expands with increased risk aversion. 

\subsubsection{Risk Dominance}

An important classical notion of comparative riskiness of actions is \textit{risk dominance}. In symmetric (\(2 \times 2\)) coordination games with two pure-strategy Nash equilibria, one equilibrium is risk dominant if and only if the action players coordinate on is optimal under a uniform prior over the other player's action. Equivalently, in a risk-dominant equilibrium, the players play the action that is optimal for \textit{the majority} (\(\geq 50\)\%) of beliefs.

The risk-dominance literature attempts to determine which of multiple equilibria is more robust to strategic uncertainty by providing a selection criterion. The seminal result of \cite{carlsson1993global} is that in \(2 \times 2\) global games (with two strict Nash equilibria) vanishing noise yields unique selection of the risk-dominant equilibrium by iterated deletion of dominant strategies.\footnote{See also, e.g., \cite{morris1995p,morris1998unique}.} It does not, however, tell us whether this equilibrium selection is robust to changes in risk aversion. 

Consider the following symmetric \(2\times 2\) game with actions \(\{a,b\}\) and the following (monetary) payoffs.
\[
\begin{array}{c|cc}
 & a & b\\\hline
a & (\alpha,\alpha) & (\beta,\gamma)\\
b & (\gamma,\beta) & (\delta,\delta)
\end{array}
\]
Let $\alpha=16$, $\delta=8$, $\beta=0$, and $\gamma=7$.
Risk dominance compares the expected payoffs if player \(1\) believes that player \(2\) plays each action with probability \(1/2\). If agents are risk neutral, then the expected payoff from \(a\) is 8, while the expected payoff from \(b\) is \(7.5\), making equilibrium \((a,a)\) risk dominant. If, however, players' utilities are $u(x) = \sqrt{x}$, then the expected utility from \(a\) is 2, while from \(b\) it is \(2.737\), so that \((b,b)\) is now risk-dominant. 

Safety identifies precisely when this type of reversal cannot occur. Again consider the game above with either $\alpha > \gamma$ and $\delta > \beta$, or $\alpha < \gamma$ and $\delta < \beta$ so that there are two strict equilibria. Now, however, adjust the payoffs so that \(a \succeq_S b\). If \((a,a)\) is risk-dominant when players have some utility \(u\), then \((a,a)\) remains risk-dominant for any increase in players' risk aversion. In other words, safety provides a condition on payoffs under which equilibrium selection is robust to increased concavity of preferences. 

A natural conjecture--related to \cite{hellwig2002imperfect}--is that in such games the selection result can be extended to allowing two dimensions of uncertainty: both about the material payoffs and the players' attitudes toward risk.\footnote{See also \cite{cabrales2007equilibrium,heinemann2009measuring}.} In fact, it is easy to characterize exactly when two strict Nash equilibria plus risk-dominance imply safer-than. Recall that in order to have two strict Nash equilibria, we need 
\[
\left(\alpha>\gamma \text{ and } \beta<\delta\right)\ \text{or}\ \left(\alpha<\gamma \text{ and } \beta>\delta\right).
\]
Next note that if action \(a\) is risk dominant, it must be that \(\alpha+\beta\geq \gamma+\delta\). Then, the implication
\[
\left[\left(\alpha+\beta\geq \gamma+\delta\right)\ \text{and}\ \left(\left(\alpha>\gamma \text{ and } \beta<\delta\right)\ \text{or}\ \left(\alpha<\gamma \text{ and } \beta>\delta\right)\right)\right]\Rightarrow \left(a\succeq_S b\right)
\]
holds if and only if
\[
\left[\left(\alpha>\gamma \text{ and } \beta<\delta\right)\Rightarrow \delta\geq \alpha\right]\ \text{and}\ \left[\left(\alpha<\gamma \text{ and } \beta>\delta\right)\Rightarrow \gamma\geq \beta\right].
\]

\subsection{Safe Securities}\label{security}

Consider a firm that is selling state-contingent securities to investors in order to raise capital. In state $\theta \in [0,1]$, the firm receives cash flow \(\theta\) and pays the promised security $S(\theta)$ to the investor, with the distribution of states denoted $\mu \in \Delta([0,1])$. These securities can be structured in a variety of ways, and many papers in the literature examine which type of security is best for either the firm to sell or the investor to offer, given the specific context. We use the tools that we have developed to compare the robustness to risk aversion of securities. 

While there is an abundance of ways to structure the state-contingent payoff $S(\theta)$, here are three of the most common:
\begin{enumerate}[noitemsep,topsep=0pt]
    \item Equity: the investor receives a constant portion, $\eta \in (0,1)$, of the cash flow, or $S(\theta) = \eta \theta$,
    \item Debt: the investor is owed a debt $d \in (0,1)$ and collects as much of it as possible, or $S(\theta) = \min\{\theta,d\}$; and 
    \item Call Option: the investor gets a call option with a strike price of $\rho \in (0,1)$, or $S(\theta) = \max\{\theta - \rho, 0 \}$. 
\end{enumerate}

Following the literature--see, e.g., \cite*{nachman1994optimal} and \cite*{demarzo2005bidding}-- we assume that any security $S(\theta)$ satisfies the following properties:
\begin{enumerate}[label={(\roman*)},noitemsep,topsep=0pt]
    \item Monotonicity I: \(S\) (the investor's share of cash flow) is nondecreasing in \(\theta\);
    \item Monotonicity II: \(\theta - S\) (the firm's share of cash flow) is nondecreasing in \(\theta\);
    \item Limited liability: \(0 \leq S\left(\theta\right) \leq \theta\).
\end{enumerate}

Take two securities, \(S_a\) and \(S_b\), and assume that neither is weakly dominated: there exist realizations of the random cash flow under which each is strictly preferred (\textit{ex post}) to the other. By Theorem \ref{thm:central}, security \(S_a\) is safer than \(S_b\) if and only if for all \(\theta \in \mathcal{A}\) and \(\theta' \in \mathcal{B}\), \(S_b\left(\theta'\right) \geq S_a\left(\theta\right)\) and \(S_a\left(\theta'\right) \geq S_b\left(\theta\right)\). Recalling our definition of single-crossing from \S\ref{sec:singlecross}, we have
\begin{proposition}\label{securitiessafety}
    \(S_a \succeq_S S_b\) if and only if \(S_b \succeq_{sc} S_a\).
\end{proposition}
\begin{proof}
    \(\left(\Leftarrow\right)\) Let \(S_b \succeq_{sc} S_a\). Then for all \(\theta \in \mathcal{A}\) and \(\theta' \in \mathcal{B}\),
    \(S_b\left(\theta'\right) > S_a\left(\theta'\right) \geq S_a\left(\theta\right)\)
    and
    \(S_a\left(\theta'\right) \geq S_a\left(\theta\right) > S_b\left(\theta\right) \),
    where we used \(S_a\)'s monotonicity.

    \bigskip

    \noindent \(\left(\Rightarrow\right)\) Suppose for the sake of contraposition \(S_b \not\succeq_{sc} S_a\). This means there exist \(\theta_1 < \theta_2\) such that \(S_a\left(\theta_1\right) < S_b\left(\theta_1\right)\) and \(S_b\left(\theta_2\right) < S_a\left(\theta_2\right)\). By the monotonicity of \(S_b\), \(S_b\left(\theta_1\right) \leq S_b\left(\theta_2\right) < S_a\left(\theta_2\right)\), so \(S_a \not\succeq_{S} S_b\).\end{proof}
    
    Debt yields the highest possible payoff for low states ($S(\theta) = \theta$), and as such, debt must single-cross any other security from above. Thus,
    \begin{corollary}
        Debt is safer than any other security.
    \end{corollary}
    Intuitively, a security is safer if it gives a constant payoff, regardless of the state. This is constrained, however, by limited liability ($S(\theta) \leq \theta$). A call option yields the lowest possible payoff for low states ($S(\theta) = 0$), and has the largest possible slope when it is non-constant. Consequently, by the intermediate value theorem, a call option must single-cross from below any other security. So,
    \begin{corollary}
        All securities are safer than a call option.
    \end{corollary}
    As observed by, for instance, \cite*{demarzo2005bidding}, \cite*{dang2013information}, and \cite*{inostroza2022optimal}, a security that crosses the other from below is more information sensitive: the value of information about the state is higher. We show that this precise means of comparison (signal-crossing from below), is the way to rank securities in terms of their robustness to risk aversion.

\subsection{Insurance}\label{ss:insurance}

Just as we ranked securities in terms of their robustness to risk aversion, we can do the same with insurance. Suppose that an agent incurs a random loss \(w \in \left[-1,0\right]\), so we define the state to be \(\theta = -w\). We define an (indemnity-schedule)\footnote{These definitions are borrowed from Definitions 6 and 8 in \cite{maccheroniinsurancerisk}.} insurance policy to be a nondecreasing function of \(\theta\), \(\iota\left(\theta\right)\), such that \(\iota\left(\theta\right) - \theta\) is nonincreasing in \(\theta\). Given a purchase of insurance, the agent's random monetary payoff is \(\iota(\theta) - \theta\).

Here are three common varieties of insurance. In each case \(p > 0\).
\begin{enumerate}[noitemsep,topsep=0pt]
    \item Full Insurance: the agent incurs no risk, i.e., $\iota(\theta) = \theta - p$.
    \item Proportional Insurance: the agent is reimbursed a constant fraction of her loss, i.e., $\iota(\theta) = (1-\eta) \theta - p$, with \(\eta \in \left[0,1\right)\).
    \item Deductible-limit Insurance: the agent is reimbursed all except deductible $\delta$, but only up to a limit $\lambda$, i.e., $\iota(\theta) = \min\left\{\left(\theta - \delta\right)^+,\lambda\right\} - p$, where \(\delta \in \mathbb{R}\) and \(\lambda \geq 0\).
\end{enumerate}
Consider two insurance policies, \(\iota_a\) and \(\iota_b\), neither of which dominates the other.  By Theorem \ref{thm:central}, insurance \(\iota_a\) is safer than \(\iota_b\) if and only if for all \(\theta \in \mathcal{A}\) and \(\theta' \in \mathcal{B}\), \(\iota_b\left(\theta'\right) - \theta' \geq \iota_a\left(\theta\right) - \theta\) and \(\iota_a\left(\theta'\right) - \theta' \geq \iota_b\left(\theta\right) - \theta\).

Recalling our definition of single-crossing from \S\ref{sec:singlecross}, we have
\begin{proposition}
    \(\iota_a \succeq_S \iota_b\) if and only if \(\iota_a(\theta) - \theta \succeq_{sc} \iota_b(\theta) - \theta\).
\end{proposition}
\begin{proof}
    This follows immediately from Proposition \ref{securitiessafety}, as the DM's payoff as a function of the state is formally equal to a security except it is now nonincreasing in the state rather than nondecreasing.
\end{proof}

A trivial corollary of this proposition is that full insurance is safer than any other insurance policy. Interestingly, any ranking of proportional and deductible-limit insurance according to safety is possible: proportional may be safer than deductible-limit or vice-versa, or they may be incomparable.

\subsection{Hedging}\label{hedge}
    Now we examine the classic question of constructing an asset portfolio to hedge against risk. To be concrete, suppose our DM's wealth, or other holdings, \(y\), fluctuates according to market conditions (the state \(\theta\)) and is distributed according to \(H_\theta\). She decides whether to add asset \(a\) or asset \(b\) to her portfolio to hedge against the unavoidable variation of \(y\). Asset \(a\) pays \(a_\theta\) in state \(\theta\) and asset \(b\) pays \(b_\theta\).

    We say that \emph{asset \(a\) hedges risk better than asset \(b\)} if any strongly risk-averse DM--whose preferences satisfy the other assumptions in \S\ref{sec:beyond}--prefers asset \(a\) to \(b\) whenever a risk-neutral DM does. Note that we choose to return to the set up of \S\ref{sec:beyond} and examine the question of hedging outside of the expected utility framework. As will become clear in our discussion after Proposition \ref{hedgeproof}, this allows us to have stronger results with fewer assumptions than previous literature. 

    For each belief \(\mu\in\Delta\), let \(X^{\mu}_{a}\) and \(X^{\mu}_{b}\) denote the aggregate wealth lotteries induced by adding assets \(a\) and \(b\), respectively: \(\theta\) is drawn according to \(\mu\), then \(y\) is drawn according to \(H_\theta\), and \(X^{\mu}_{a}=a_\theta+y\) while \(X^{\mu}_{b}=b_\theta+y\). Let \(\mathcal{A}\coloneqq\left\{\theta\in\Theta\colon a_\theta>b_\theta\right\}\), \(\mathcal{B}\coloneqq\left\{\theta\in\Theta\colon b_\theta>a_\theta\right\}\), and \(\mathcal{C}\coloneqq\left\{\theta\in\Theta\colon a_\theta=b_\theta\right\}\). For each \(\left(\theta,\theta'\right)\in\mathcal{A}\times\mathcal{B}\), define
    \[
        \lambda_{\theta,\theta'}\coloneqq\frac{a_\theta-b_\theta}{\left(a_\theta-b_\theta\right)+\left(b_{\theta'}-a_{\theta'}\right)}\in\left(0,1\right)\text{,}
    \]
    and let \(\mu_{\theta,\theta'}\in\Delta\) denote the belief supported on \(\left\{\theta,\theta'\right\}\) such that \(\mu_{\theta,\theta'}\left(\theta'\right)=\lambda_{\theta,\theta'}\). Thus,
    \begin{proposition}\label{hedgeproof}
        Asset \(a\) hedges risk better than asset \(b\) if and only if for all \(\left(\theta,\theta'\right)\in\mathcal{A}\times\mathcal{B}\), the aggregate wealth lottery \(X^{\mu_{\theta,\theta'}}_{a}\) SOSD \(X^{\mu_{\theta,\theta'}}_{b}\).
    \end{proposition}
    
    Proposition \ref{hedgeproof} provides an exact characterization of robust hedging in terms of a collection of two-state SOSD comparisons. The classical literature--e.g., \cite*{kihlstrom1981risk}, \cite*{ross1981some}, \cite*{nachman1982preservation}, \cite*{jewitt1987risk}, and \cite*{eeckhoudt1996changes}--studies a similar question, asking when \[\mathbb{E}_{F^{\mu}_{a},H}u\left(Y_a+V\right)\ge\mathbb{E}_{F^{\mu}_{b},H}u\left(Y_b+V\right)\quad\Rightarrow\quad\mathbb{E}_{F^{\mu}_{a},H}\phi\circ u\left(Y_a+V\right)\ge\mathbb{E}_{F^{\mu}_{b},H}\phi\circ u\left(Y_b+V\right)\text{,}\]
for any monotone concave \(\phi\), where random variable \(V\) (with cdf \(H\)) is the DM's random wealth and \(Y_a\) and \(Y_b\) (with respective cdfs \(F^{\mu}_{a}\) and \(F^{\mu}_{b}\)) are the random payoffs from assets \(a\) and \(b\), respectively. 

Even restricting to the expected-utility case, our result is different from these works in three essential ways: first, the aforementioned papers are concerned with the properties of random variables that lead to one asset being more sensitive to an investor's risk preferences. In contrast, ours is a distribution-free comparison between assets. Second, the papers in the literature stipulate that \(Y_i\) and \(V\) are independent. We assume no such independence here. Third, our robust approach allows for an equivalence result, while the results in the literature provide sufficient conditions. Of course, the preferences valid for our finding are also much more general than just expected utility.

\section{Discussion}

This paper contributes to the sizable literature evaluating decision-making under risk, by studying the actions in a decision problem when the uncertainty is subjective (or unknown). That is, rather than fixing a belief and comparing the induced lotteries over outcomes--as is standard (at least implicitly) in the literature--we ask when one action becomes more attractive than another across an expanding set of beliefs as the decision-maker becomes more risk averse.

Our main result equates this risk-averse-robust property with a simple collection of pairwise conditions on state-dependent payoffs as well as to a robust form of single-crossing. This begs the question: what is the benefit of our approach, when one could instead check the classical (single-crossing) condition pointwise, i.e., for every belief? Our answer is that, in a sense, we are doing exactly this: the linearity of expected utility means that the behavior of a DM can be summarized by her behavior at particular extreme beliefs, the binary distributions at which she is indifferent between the actions. Our main result, therefore, could be seen as equating this binary-state single-crossing condition with a simple convex-hull condition on payoffs.

Moreover, our robust approach offers significant advantages over the classical known-lotteries studies: our characterization requires knowledge only of the state-dependent monetary payoffs of actions. This means that our results readily apply to environments in which beliefs (uncertainty) are unspecified (even misspecified) or heterogeneous. For example, a principal selling insurance or asset-backed securities is surely able to specify the terms of the contract--the state-dependent payoffs--but may not know consumers' beliefs or specific risk-preferences beyond a lower or upper bound. Likewise, safety provides a tractable and meaningful way to conduct belief- and utility-robust comparative statics \citep{scs}.

\bibliography{sample.bib}

@article{heinemann2009measuring,
  title={Measuring strategic uncertainty in coordination games},
  author={Heinemann, Frank and Nagel, Rosemarie and Ockenfels, Peter},
  journal={The Review of Economic Studies},
  volume={76},
  number={1},
  pages={181--221},
  year={2009}}

@article{cabrales2007equilibrium,
  title={Equilibrium selection through incomplete information in coordination games: an experimental study},
  author={Cabrales, Antonio and Nagel, Rosemarie and Armenter, Roc},
  journal={Experimental Economics},
  volume={10},
  number={3},
  pages={221--234},
  year={2007}}

@article{hellwig2002imperfect,
  title={Imperfect common knowledge of preferences in global coordination games},
  author={Hellwig, Christian},
  journal={Mimeo},
  year={2002}}

@book{shaked2007stochastic,
  title={Stochastic orders},
  author={Shaked, Moshe and Shanthikumar, J George},
  year={2007},
  publisher={Springer}
}

@article{szpilrajn1930extension,
  title={Sur l'extension de l'ordre partiel},
  author={Szpilrajn, Edward},
  journal={Fundamenta mathematicae},
  volume={1},
  number={16},
  pages={386--389},
  year={1930}
}

@article{curelloetal,
    author = {Curello, Gregorio and Sinander, Ludvig and Whitmeyer, Mark},
    title = {Comparative risk attitude and the aggregation of single-crossing},
    journal = {Mimeo},
    year = {2025}
}

@article{carlsson1993global,
  title={Global games and equilibrium selection},
  author={Carlsson, Hans and Van Damme, Eric},
  journal={Econometrica: Journal of the Econometric Society},
  pages={989--1018},
  year={1993}}

@article{morris1998unique,
  title={Unique equilibrium in a model of self-fulfilling currency attacks},
  author={Morris, Stephen and Shin, Hyun Song},
  journal={American Economic Review},
  pages={587--597},
  year={1998}}

@article{morris1995p,
  title={p-Dominance and belief potential},
  author={Morris, Stephen and Rob, Rafael and Shin, Hyun Song},
  journal={Econometrica: Journal of the Econometric Society},
  pages={145--157},
  year={1995}}

@article{kartik2023single,
  title={Single-Crossing Differences in Convex Environments},
  author={Kartik, Navin and Lee, SangMok and Rappoport, Daniel},
  journal={Review of Economic Studies},
  volume={91},
  number={5},
  pages={2981--3012},
  year={2024}
}

@article{choi2017ordinal,
  title={Ordinal aggregation results via Karlin's variation diminishing property},
  author={Choi, Michael and Smith, Lones},
  journal={Journal of Economic theory},
  volume={168},
  pages={1--11},
  year={2017}}

@article{quah2012aggregating,
  title={Aggregating the single crossing property},
  author={Quah, John K-H and Strulovici, Bruno},
  journal={Econometrica},
  volume={80},
  number={5},
  pages={2333--2348},
  year={2012}}

@article{hammond1974simplifying,
  title={Simplifying the choice between uncertain prospects where preference is nonlinear},
  author={Hammond III, John S},
  journal={Management Science},
  volume={20},
  number={7},
  pages={1047--1072},
  year={1974}
}

@article{jewitt1987risk,
  title={Risk aversion and the choice between risky prospects: the preservation of comparative statics results},
  author={Jewitt, Ian},
  journal={The Review of Economic Studies},
  volume={54},
  number={1},
  pages={73--85},
  year={1987}}

@article{karlin1963generalized,
  title={Generalized convex inequalities.},
  author={Karlin, Samuel and Novikoff, Albert},
  year={1963},
    journal={Pacific Journal of Mathematics},
    volume={13},
number={4}}

@article{lambert1979attitudes,
  title={Attitudes to risk},
  author={Lambert, Peter J and Hey, John D},
  journal={Economics Letters},
  volume={2},
  number={3},
  pages={215--218},
  year={1979}}

@article{battigalli2016note,
  title={A note on comparative ambiguity aversion and justifiability},
  author={Battigalli, Pierpaolo and Cerreia-Vioglio, Simone and Maccheroni, Fabio and Marinacci, Massimo},
  journal={Econometrica},
  volume={84},
  number={5},
  pages={1903--1916},
  year={2016}}

@book{phelps2009convex,
  title={Convex functions, monotone operators and differentiability},
  author={Phelps, Robert R},
  volume={1364},
  year={2009},
  publisher={Springer}
}

@article{weinstein2016effect,
  title={The effect of changes in risk attitude on strategic behavior},
  author={Weinstein, Jonathan},
  journal={Econometrica},
  volume={84},
  number={5},
  pages={1881--1902},
  year={2016}
}

@article{nachman1982preservation,
  title={Preservation of “more risk averse” under expectations},
  author={Nachman, David C},
  journal={Journal of Economic Theory},
  volume={28},
  number={2},
  pages={361--368},
  year={1982}
}

@article{ross1981some,
  title={Some stronger measures of risk aversion in the small and the large with applications},
  author={Ross, Stephen A},
  journal={Econometrica},
  volume={49},
  number={3},
  pages={621--638},
  year={1981}
}

@article{eeckhoudt1996changes,
   title={Changes in background risk and risk taking behavior},
  author={Eeckhoudt, Louis and Gollier, Christian and Schlesinger, Harris},
  journal={Econometrica},
  volume={64},
  number={3},
  pages={683--689},
  year={1996}}

@article{kihlstrom1981risk,
  title={Risk aversion with random initial wealth},
  author={Kihlstrom, Richard E and Romer, David and Williams, Steve},
  journal={Econometrica},
  volume={49},
  number={4},
  pages={911--920},
  year={1981}}

@article{klibanoff2005smooth,
  title={A smooth model of decision making under ambiguity},
  author={Klibanoff, Peter and Marinacci, Massimo and Mukerji, Sujoy},
  journal={Econometrica},
  volume={73},
  number={6},
  pages={1849--1892},
  year={2005}}

@article{maccheroniinsurancerisk,
  title={Risk aversion and insurance propensity},
  author={Maccheroni, Fabio and Marinacci, Massimo and Wang, Ruodu and Wu, Qinyu},
  journal={American Economic Review},
  volume={115},
  number={5},
  pages={1597--1649},
  year={2025}}

@article{borie2023maxmin,
  title={Maxmin expected utility in Savage's framework},
  author={Borie, Dino},
  journal={Journal of Economic Theory},
  volume={210},
  pages={105665},
  year={2023}}

@article{bewley2002knightian,
  title={Knightian decision theory. Part I},
  author={Bewley, Truman F},
  journal={Decisions in economics and finance},
  volume={25},
  pages={79--110},
  year={2002}}

@article{ghirardato2003subjective,
  title={A subjective spin on roulette wheels},
  author={Ghirardato, Paolo and Maccheroni, Fabio and Marinacci, Massimo and Siniscalchi, Marciano},
  journal={Econometrica},
  volume={71},
  number={6},
  pages={1897--1908},
  year={2003}}

@article{demarzo2005bidding,
  title={Bidding with securities: Auctions and security design},
  author={DeMarzo, Peter M and Kremer, Ilan and Skrzypacz, Andrzej},
  journal={American economic review},
  volume={95},
  number={4},
  pages={936--959},
  year={2005}}

@article{inostroza2022optimal,
  title={Optimal information and security design},
  author={Inostroza, Nicolas and Tsoy, Anton},
  journal={Mimeo},
  year={2022}
}

@article{yaari69,
  title={Some remarks on measures of risk aversion and on their uses},
  author={Yaari, Menahem E},
  journal={Journal of Economic Theory},
    volume={1},
  pages={315-329},
  year={1969},
    number={3}}

@article{ghirardato2002ambiguity,
  title={Ambiguity made precise: A comparative foundation},
  author={Ghirardato, Paolo and Marinacci, Massimo},
  journal={Journal of Economic Theory},
  volume={102},
  number={2},
  pages={251--289},
  year={2002}}

@article{dang2013information,
  title={The information sensitivity of a security},
  author={Dang, Tri Vi and Gorton, Gary and Holmstr{\"o}m, Bengt},
  journal={Mimeo},
  pages={39--65},
  year={2013}
}

@article{nachman1994optimal,
  title={Optimal design of securities under asymmetric information},
  author={Nachman, David C and Noe, Thomas H},
  journal={The Review of Financial Studies},
  volume={7},
  number={1},
  pages={1--44},
  year={1994}}

@ARTICLE{scs,
   author = {Whitmeyer, Mark},
    title = {Comparative Statics for the Subjective},
  journal = {Mimeo},
     year = 2025,
    month = Jan,
}

@article{peasewhitmeyerauctions,
Author = {Pease, Marilyn and Whitmeyer, Mark},
Title = {On Risk Aversion in Auctions},
Journal = {Mimeo},
Year = {2026},
Month = {February}}

@article{flexibilitypaper,
  title={Making Information More Valuable},
  author={Whitmeyer, Mark},
  journal={Mimeo},
  year={2023}
}

@article{milgrom1994monotone,
  title={Monotone comparative statics},
  author={Milgrom, Paul and Shannon, Chris},
  journal={Econometrica},
  volume={62},
  number={1},
  pages={157--180},
  year={1994}}

@article{athey2002monotone,
  title={Monotone comparative statics under uncertainty},
  author={Athey, Susan},
  journal={The Quarterly Journal of Economics},
  volume={117},
  number={1},
  pages={187--223},
  year={2002}}

@article{edlin1998strict,
  title={Strict single crossing and the strict Spence-Mirrlees condition: a comment on monotone comparative statics},
  author={Edlin, Aaron S and Shannon, Chris},
  journal={Econometrica},
  volume={66},
  number={6},
  pages={1417--1425},
  year={1998}}

@article{ROTHSCHILD1970225,
title = {Increasing risk: I. A definition},
journal = {Journal of Economic Theory},
volume = {2},
number = {3},
pages = {225-243},
year = {1970},
author = {Michael Rothschild and Joseph E Stiglitz}
}

@article{aumann2008economic,
  title={An economic index of riskiness},
  author={Aumann, Robert J and Serrano, Roberto},
  journal={Journal of Political Economy},
  volume={116},
  number={5},
  pages={810--836},
  year={2008},
  publisher={The University of Chicago Press}
}

@article{foster2009operational,
  title={An operational measure of riskiness},
  author={Foster, Dean P and Hart, Sergiu},
  journal={Journal of Political Economy},
  volume={117},
  number={5},
  pages={785--814},
  year={2009},
  publisher={The University of Chicago Press}
}

@article{jewitt2017ordering,
  title={Ordering ambiguous acts},
  author={Jewitt, Ian and Mukerji, Sujoy},
  journal={Journal of Economic Theory},
  volume={171},
  pages={213--267},
  year={2017},
  publisher={Elsevier}
}

@article{palfrey1984participation,
  title={Participation and the provision of discrete public goods: a strategic analysis},
  author={Palfrey, Thomas R and Rosenthal, Howard},
  journal={Journal of public Economics},
  volume={24},
  number={2},
  pages={171--193},
  year={1984},
  publisher={Elsevier}
}

@article{palfrey1985voter,
  title={Voter participation and strategic uncertainty},
  author={Palfrey, Thomas R and Rosenthal, Howard},
  journal={American political science review},
  volume={79},
  number={1},
  pages={62--78},
  year={1985},
  publisher={Cambridge University Press}
}

@article{bagnoli1989provision,
  title={Provision of public goods: Fully implementing the core through private contributions},
  author={Bagnoli, Mark and Lipman, Barton L},
  journal={The Review of Economic Studies},
  volume={56},
  number={4},
  pages={583--601},
  year={1989},
  publisher={Wiley-Blackwell}
}

@article{teyssier2012inequity,
  title={Inequity and risk aversion in sequential public good games},
  author={Teyssier, Sabrina},
  journal={Public Choice},
  volume={151},
  number={1},
  pages={91--119},
  year={2012},
  publisher={Springer}
}

\appendix

\section{Omitted Proofs}\label{omittedproofs}

    \subsection{Theorem \ref{thm:central} Proof}\label{thm:centralproof}

To complete the proof, we prove two lemmas; Lemma \ref{lem:1equiv2}, which establishes \ref{bullet1} \(\Longleftrightarrow\) \ref{bullet2}; and Lemma \ref{lem:sconlyif} which gives us \ref{bullet2} \(\Longrightarrow\) \ref{bullet3}. 

\begin{lemma}\label{lem:1equiv2}
    Action \(a\) is safer than action \(b\) if and only if for each \(\theta \in \mathcal{A}\) and \(\theta' \in \mathcal{B}\), \(b_{\theta'} \geq a_{\theta}\) and \(a_{\theta'} \geq b_{\theta}\).
\end{lemma}
\begin{proof}
We start by constructing a number of objects generated when the DM's utility is \(u\), understanding their analogs for utility \(\hat{u} = \phi \circ u\) to be generated in the same manner.

We extend our two-state convention and write \(\alpha_\theta \coloneqq u\left(a_\theta\right)\) and  \(\beta_\theta \coloneqq u\left(b_\theta\right)\). Let \(H_{a,b}\) denote the hyperplane of indifference, the set of beliefs at which the DM is indifferent between \(a\) and \(b\): \[H_{a,b} \coloneqq \left\{\mu \in \Delta \colon \mathbb{E}_{\mu}\alpha_\theta = \mathbb{E}_{\mu}\beta_\theta\right\}\text{.}\]

For any \(\theta \in \Theta\), let \(v_{\theta}\) denote the corresponding vertex (as a point in the simplex). Furthermore, for any \(\theta \in \mathcal{A}\) and \(\theta' \in \mathcal{B}\), let \(e_{\theta, \theta'}\) denote the edge of \(\Delta\) ``between the two states:'' \[e_{\theta,\theta'} \coloneqq \left\{\mu \in \Delta \vert \exists \lambda \in \left(0,1\right) \colon \lambda v_{\theta} + \left(1-\lambda\right) v_{\theta'} = \mu\right\}\text{,}\] and let \(\mu_{\theta,\theta'}\) denote the point of indifference between the two actions lying on the edge \(e_{\theta,\theta'}\):
 \[\mu_{\theta,\theta'} \coloneqq \left\{\mu \in e_{\theta,\theta'}  \colon \mathbb{E}_{\mu}\alpha_\theta = \mathbb{E}_{\mu}\beta_\theta\right\}\text{.}\]
 Equivalently, \(\mu_{\theta,\theta'} = e_{\theta,\theta'} \cap H_{a,b}\).

\(P_{a,b}\left(a\right)\) is a convex set, so by the Krein-Milman theorem, it is the closed convex hull of its extreme points. Furthermore, its set of extreme points is the set \[\ubar{P} \coloneqq \left\{v_\theta \colon \theta \in \mathcal{A} \cup \mathcal{C}\right\} \cup \left\{\mu_{\theta,\theta'} \colon \theta \in \mathcal{A}, \ \theta' \in \mathcal{B}\right\}\text{,}\] i.e., the set of vertices at which \(a\) is uniquely optimal, and the points on the edges connecting the vertices at which \(a\) is uniquely optimal with the vertices at which \(b\) is uniquely optimal.

Now, \(P_{a,b}\left(a\right) \subseteq \hat{P}_{a,b}\left(a\right)\) if and only if \(\ubar{P} \subseteq \bar{\co} \hat{\ubar{P}}\), i.e., the set \(\ubar{P}\) lies within the closed convex hull of the set \(\hat{\ubar{P}}\). Moreover, by construction \(\ubar{P} \subseteq \bar{\co} \hat{\ubar{P}}\) if and only if for each \(\theta \in \mathcal{A}\) and \(\theta' \in \mathcal{B}\), there exists some \(\lambda \in \left(0,1\right]\) such that \(\mu_{\theta,\theta'} = \lambda \hat{\mu}_{\theta,\theta'} + \left(1-\lambda\right) v_{\theta}\). That is, for each edge containing a point of indifference, the indifference point for utility \(u\), \(\mu_{\theta,\theta'}\), must lie between the indifference point in decision problem for utility \(\hat{u}\) and the vertex \(v_{\theta}\).

By Proposition \ref{twostatesaferthanequiv}, this holds for any strictly monotone concave transformation of \(u\), \(\phi\), if and only if for all \(\theta \in \mathcal{A}\) and \(\theta' \in \mathcal{B}\), \(\beta_{\theta'} \geq \alpha_{\theta}\) and \(\alpha_{\theta'} \geq \beta_{\theta}\). \end{proof}

\begin{lemma}\label{lem:sconlyif}
If \(a \succeq_S b\), then \(F^{\mu}_{a} \succeq_{sc} F^{\mu}_{b}\) for all \(\mu \in \Delta\).
\end{lemma}
\begin{proof}
From \ref{bullet2} (for any pair \((\theta, \theta') \in \mathcal{A} \times \mathcal{B}\), \(b_{\theta'} \geq a_\theta\) and \(a_{\theta'} \geq b_\theta\)), \(\sup_{\theta \in \mathcal{A}} b_\theta \leq \inf_{\theta' \in \mathcal{B}} a_{\theta'} \leq \inf_{\theta' \in \mathcal B} b_{\theta'}\), and \(\sup_{\theta \in \mathcal{A}} a_\theta \leq \inf_{\theta' \in \mathcal{B}} b_{\theta'}\). Moreover, we recall that \(a_\theta > b_\theta\) for all \(\theta \in \mathcal{A}\), \(a_{\theta'} < b_{\theta'}\) for all \(\theta' \in \mathcal{B}\), and \(a_{\theta^{\dagger}} = b_{\theta^{\dagger}}\) for all \(\theta^{\dagger} \in \mathcal{C}\). 

Fix \(\mu \in \Delta\). For \(t \in \mathbb{R}\), we define \(D(t) \coloneqq F_a^{\mu}(t) - F_b^{\mu}(t)\). We also define \(\ubar{a} \coloneqq \inf_{\theta \in \mathcal{B}} a_\theta\) and \(\bar{a} \coloneqq \sup_{\theta \in \mathcal{A}} a_\theta\).

First, for any \(t < \ubar{a}\), we have \(\mu\left(\left\{\left.\theta \right| a_\theta \leq t\right\}\right) = \mu\left(\left\{\left.\theta \in \mathcal{A}\cup \mathcal{C} \right| a_\theta \leq t\right\}\right)\), since \(a_\theta \geq \ubar{a}\) for all \(\theta \in \mathcal{B}\).
On \(\mathcal{A}\cup\mathcal{C}\), we have \(b_\theta \leq a_\theta\), so
\(\left\{\left.\theta \in \mathcal{A}\cup\mathcal{C} \right| a_\theta \leq t\right\} \subseteq
\left\{\left.\theta \right| b_\theta \leq t\right\}\).
Hence, \(F_b^{\mu}(t) \geq F_a^{\mu}(t)\), i.e., \(D(t) \leq 0\), for all \(t < \ubar{a}\).

Second, for any \(t \geq \bar{a}\), we claim that
\(\left\{\left.\theta \right| b_\theta \leq t\right\} \subseteq \left\{\left.\theta \right| a_\theta \leq t\right\}\).
Indeed, if \(\theta \in \mathcal{B}\cup\mathcal{C}\), then \(a_\theta \leq b_\theta\), so \(b_\theta \leq t\) implies \(a_\theta \leq t\). If \(\theta \in \mathcal{A}\), then \(a_\theta \leq \bar{a} \leq t\).
Thus \(F_a^{\mu}(t) \geq F_b^{\mu}(t)\), i.e., \(D(t) \geq 0\), for all \(t \geq \bar{a}\).

Finally, \(\sup_{\theta \in \mathcal{A}} b_\theta \leq \ubar{a}\) and \(\bar{a} \leq \inf_{\theta \in \mathcal{B}} b_\theta\).
Therefore, on the open interval \(\left(\ubar{a},\bar{a}\right)\), the function \(D(\cdot)\) is weakly increasing. It follows that \(D(\cdot)\) crosses \(0\) at most once (and from below), which is exactly \(F_a^{\mu} \succeq_{sc} F_b^{\mu}\).\end{proof}

\subsection{Proposition \ref{twostatesaferthanequiv} Proof}\label{proofsafe}
\begin{proof} \(\left(\Rightarrow\right)\) By the strict monotonicity of \(u\) we may work directly with utils rather than outcomes; \textit{viz.,} \(\alpha\)s and \(\beta\)s rather than \(a\)s and \(b\)s. 

Suppose for the sake of contraposition that \(\beta_0 > \alpha_1\) (and recall \(\alpha_0 > \beta_0\)). There are two possibilities: either \(\alpha_0 \leq \beta_1\), or \(\alpha_0 > \beta_1\). 
    
    Suppose first \(\alpha_0 \leq \beta_1\), so \(\beta_1 \geq \alpha_0 > \beta_0 > \alpha_1\) and let \(\phi\left(y\right) = \min\left\{y, k y + c\right\}\),
    where 
    \[c = \frac{\beta_0 \left(\beta_0\beta_1 - \alpha_0 \alpha_1\right)}{\beta_0 \left(\alpha_1 - \beta_0\right)+ \alpha_0 \left(\beta_0 -2\alpha_1\right)+ \beta_0\beta_1} \ \text{and} \ k = \frac{\left(\alpha_0 -\beta_0\right)\left(\beta_0 - \alpha_1\right)}{\beta_0 \left(\alpha_1 - \beta_0\right)+ \alpha_0 \left(\beta_0 -2\alpha_1\right)+ \beta_0\beta_1}\text{.}\]
    It is straightforward to check that \(k \in \left(0,1\right)\) and \(k \beta_0 + c = \beta_0\), so \(\phi\) is weakly concave, as required. Moreover, 
    \[\frac{\phi\left(\beta_1\right) - \phi\left(\alpha_1\right)}{\beta_1 - \alpha_1} > \frac{\phi\left(\alpha_0\right) - \phi\left(\beta_0\right)}{\alpha_0 - \beta_0} \ \Leftrightarrow \ \frac{k \beta_1 + c - \alpha_1}{\beta_1-\alpha_1} - \frac{k \alpha_0 + c - \beta_0}{\alpha_0-\beta_0} > 0 \text{,}\]
    if and only if
    \[(1-k)\left(\beta_0 \beta_1 - \alpha_0 \alpha_1\right) - c \left(\beta_1 - \alpha_0 + \beta_0 - \alpha_1\right) > 0\text{,}\]
    which also holds.

    Finally, suppose \(\alpha_0 > \beta_1\), in which case we have
    \( \alpha_0 > \beta_1 > \alpha_1\) and \(\alpha_0 > \beta_0 > \alpha_1\). By the three-chord lemma, we have
    \[\frac{\phi\left(\beta_1\right) - \phi\left(\alpha_1\right)}{\beta_1 - \alpha_1} \geq \frac{\phi\left(\alpha_0\right) - \phi\left(\alpha_1\right)}{\alpha_0 - \alpha_1}\text{,}\]
    so it suffices to construct a concave monotone \(\phi\) for which
    \[\Psi \coloneqq \frac{\phi\left(\alpha_0\right) - \phi\left(\alpha_1\right)}{\alpha_0 - \alpha_1} - \frac{\phi\left(\alpha_0\right) - \phi\left(\beta_0\right)}{\alpha_0 - \beta_0} > 0\text{.}\] To that end, let \[\phi\left(y\right) = \min\left\{y, \frac{y + \beta_0}{2}\right\}\text{.}\]
    Plugging this in, we have
    \[\Psi = \frac{\frac{\alpha_0 + \beta_0}{2} - \alpha_1}{\alpha_0 - \alpha_1} -  \frac{\frac{\alpha_0 + \beta_0}{2} - \beta_0}{\alpha_0 - \beta_0} = \frac{\beta_0 - \alpha_1}{2\left(\alpha_0-\alpha_1\right)} > 0\text{,}\] as desired.\end{proof}

    \subsection{Corollary \ref{c:learn} Proof}

    \begin{proof}
    \((\Rightarrow)\) Immediate under consistency due to the set-inclusion ordering of beliefs corresponding to safety.

    \smallskip
    
    \noindent \((\Leftarrow)\) Suppose for the sake of contraposition that \(a\not\succeq_S b\). Then, following the proofs of Theorem \ref{thm:central} and Proposition \ref{twostatesaferthanequiv}, there exists some belief \(\mu'\) such that \(\mathbb{E}_{\mu'}\left[u(a_\theta)\right]>\mathbb{E}_{\mu'}\left[u(b_\theta)\right]\) and \(\mathbb{E}_{\mu'}\left[\hat u(a_\theta)\right]<\mathbb{E}_{\mu'}\left[\hat u(b_\theta)\right]\).

    Let
\[
\varepsilon\coloneqq\sup\left\{\eta\in(0,1)\colon\frac{\mu_0-\eta\mu'}{1-\eta}\in\Delta\right\},
\quad \text{and} \quad
\tilde{\mu}\coloneqq\frac{\mu_0-\varepsilon\mu'}{1-\varepsilon}.
\]
By the maximality of \(\varepsilon\), \(\tilde{\mu}\) lies on the boundary of \(\Delta\). Let \(\tilde{\Phi}\) be the Bayes-plausible binary distribution (over posterior beliefs) with support on \(\mu'\) and \(\tilde{\mu}\) that assigns probability \(\varepsilon\) to \(\mu'\) and \(1-\varepsilon\) to \(\tilde{\mu}\), so that \(\mathbb{E}_{\tilde{\Phi}}\mu=\varepsilon\mu'+(1-\varepsilon)\tilde{\mu}=\mu_0\). 

Let \(M\coloneqq\left\{\delta_\theta\colon\theta\in\Theta\right\}\), and let \(\Phi\) be the Bayes-plausible distribution obtained by refining \(\tilde{\Phi}\) as follows: conditional on posterior \(\tilde{\mu}\), the posterior is \(\delta_\theta\) where \(\theta\sim\tilde{\mu}\).\footnote{That is to say, let \(\iota \colon \Theta \to \Delta\) with \(\theta \mapsto_\iota \delta_\theta\). Then define \(\Phi \in \Delta \Delta(\Theta)\) by \(\Phi \coloneqq \varepsilon \delta_{\mu'} + (1-\varepsilon) (\tilde{\mu} \circ \iota^{-1})\).} Then \(\Phi\) has support on \(\left\{\mu'\right\}\cup M\) and satisfies
\[
\mathbb{E}_\Phi\mu=\varepsilon\mu'+(1-\varepsilon)\int_\Theta\delta_\theta\,d\tilde{\mu}(\theta)=\varepsilon\mu'+(1-\varepsilon)\tilde{\mu}=\mu_0.
\]
Under consistency and rationality, \(k(\mu)=\hat{k}(\mu)\) for all \(\mu\in M\), and \(k(\mu')=1>0=\hat{k}(\mu')\). We conclude that \(\mathbb{E}_\Phi\hat{k}(\mu)<\mathbb{E}_\Phi k(\mu)\).\end{proof}

\subsection{Proposition \ref{prop:safety_max_symm} Proof}\label{a:contribution} 
\begin{proof}
For each \(p\in\left[0,1\right]\), let \(\mu_p\) denote the induced distribution of \(M_p\) on \(\left\{0,1,\dots,n-1\right\}\).
By definition,
\[
\Delta_u(p)\ge0
\quad\Longleftrightarrow\quad
\sum_{m=0}^{n-1}\mu_p(m)\,u\left(\nu_1(m)\right)\ge \sum_{m=0}^{n-1}\mu_p(m)\,u\left(\nu_0(m)\right).
\]
As action \(1\) is safer than \(0\), the latter implication yields
\[
\sum_{m=0}^{n-1}\mu_p(m)\,\hat u\left(\nu_1(m)\right)\ge \sum_{m=0}^{n-1}\mu_p(m)\,\hat u\left(\nu_0(m)\right),
\]
i.e. \(\Delta_{\hat u}(p)\ge0\).
Thus \(\Delta_u(p)\ge0\Rightarrow \Delta_{\hat u}(p)\ge0\) for all \(p\), hence \(S_u\subseteq S_{\hat u}\).
Taking suprema delivers \(\bar{p}(\hat u)\ge \bar{p}(u)\).

It remains to justify that \(\bar{p}(u)\) is the maximal symmetric Nash equilibrium mixing probability.
A symmetric Nash equilibrium in mixing probabilities is any \(p\in\left[0,1\right]\) such that:
\[
p=0\ \Longrightarrow\ \Delta_u(0)\le0,
\qquad
p\in\left(0,1\right)\ \Longrightarrow\ \Delta_u(p)=0,
\qquad
p=1\ \Longrightarrow\ \Delta_u(1)\ge0.
\]
Let \(p^\star\coloneqq\bar{p}(u)\).

If \(S_u=\varnothing\), then \(\Delta_u(p)<0\) for all \(p\), so action \(0\) is a best reply to every symmetric conjecture and \(p=0\) is a symmetric Nash equilibrium; in this case \(p^\star=0\) by definition.

If \(S_u\ne\varnothing\) and \(p^\star=1\), then \(\Delta_u(1)\ge0\) and \(p=1\) is a symmetric Nash equilibrium.

If \(S_u\ne\varnothing\) and \(p^\star\in\left(0,1\right)\), then \(p^\star\in S_u\) and \(\Delta_u(p^\star)\ge0\).
Since \(\Delta_u(\cdot)\) is continuous (indeed a polynomial in \(p\)), we must have \(\Delta_u(p^\star)=0\); otherwise, \(\Delta_u(p^\star)>0\) would imply \(\Delta_u(p)>0\) for some \(p>p^\star\) close to \(p^\star\), contradicting the definition of \(p^\star\) as a supremum of \(S_u\). Hence, \(p^\star\) is an interior symmetric mixed equilibrium.

Finally, if \(p\) is any symmetric Nash equilibrium under \(u\), then either \(p\in\left(0,1\right)\) and \(\Delta_u(p)=0\), which implies \(p\in S_u\) and thus \(p\le\sup S_u\le p^\star\); or \(p=1\) and then \(1\in S_u\) so \(p^\star=1\); or \(p=0\) and trivially \(p\le p^\star\).
Thus \(p^\star\) is the maximal symmetric equilibrium mixing probability.
The same argument applies to \(\bar{p}(\hat u)\).
\end{proof}

\subsection{Proposition \ref{hedgeproof} Proof}\label{a:hedgeproof}
        \begin{proof}
    Let \(M_a \coloneqq \left\{\mu \in \Delta \colon \ \mathbb{E}[X_a^\mu] \geq \mathbb{E}[X_b^\mu] \right\}\). As in the proof of Theorem \ref{thm:central}, the extreme points of \(M_a\) are
    \[\mathrm{ext} M_a = \left\{\delta_{\theta}\colon \ \theta \in \mathcal{A} \cup \mathcal{C}\right\} \cup \left\{\mu_{\theta,\theta'} \colon (\theta,\theta') \in \mathcal{A} \times \mathcal{B} \wedge \mu_{\theta,\theta'}\left(\theta'\right)=\lambda_{\theta,\theta'}\right\}.\]
    For all \(\mu \in \mathrm{ext} M_a\), all expected-utility DMs prefer \(X^\mu_a\) to \(X^\mu_b\) if and only if \(X^\mu_a\) SOSD \(X^\mu_b\). This proves necessity. 
    
    For sufficiency, observe that if \(X^\mu_a\) SOSD \(X^\mu_b\) for all \(\mu \in \mathrm{ext} M_a\), all risk-averse expected-utility DMs prefer \(X^\mu_a\) to \(X^\mu_b\) for all \(\mu \in  M_a\). Consequently, \(X^\mu_a\) dominates \(X^\mu_b\) in the increasing concave order. Thus, by Theorem 4.A.6 in \citet{shaked2007stochastic}, for any \(\mu \in M_a\), there exists a random variable \(\tilde{X}^\mu\) such that \(X^\mu_a\) FOSD \(\tilde{X}^\mu\) SOSD \(X^\mu_b\), delivering the result.
    \end{proof}

\section{Generating Transitivity By Refining Safety}\label{noordersec}
    
    We can also construct refinements of safety that are transitive even when the state space is not totally ordered. Essentially, what we are doing is generating the structure we need \textit{through} the modified safety relations. To wit, let \(\left(\Theta, \underline{\blacktriangleright}\right)\) and \(\left(A, \succeq\right)\) be partially ordered sets. By Szpilrajn's extension theorem (\cite*{szpilrajn1930extension}), there exist total orders on \(\Theta\) and \(A\). We fix one of each and, abusing notation, denote them \(\underline{\blacktriangleright}\) and \(\succeq\), respectively. Then,
    \begin{definition}
        \(a\) is \emph{totally safer} than \(b\), \(a \succeq_T b\), if 
        \begin{enumerate}[label={(\roman*)},noitemsep,topsep=0pt]
            \item \(u\) is increasing in \(\theta\) for \(a\) and \(b\);
            \item \(u(b_\theta) \geq (>) u(a_\theta)\) imply \(u(b_{\theta'}) \geq (>) u(a_{\theta'})\) for all \(\theta' \underline{\blacktriangleright} \theta\); and
            \item \(a \succeq_S b\).
        \end{enumerate}
    \end{definition}
    Mirroring the proof of Proposition \ref{bigequivalence} nearly exactly, we obtain
    \begin{remark}
        \(a \succeq_T b\) if and only if \(b \geq a\). Moreover, \(\succeq_T\) is a partial order on the set of actions.
    \end{remark}
    \(\succeq_T\) is not a total order on the set of actions as not every pair of actions need satisfy our conditions for the total-safety relation. Now fix a total-order extension of the partial order on \(\Theta\), \(\underline{\blacktriangleright}\). We no longer do so with respect to \(A\).
    \begin{definition}
        \(a\) is \emph{partially safer} than \(b\), \(a \succeq_P b\), if 
        \begin{enumerate}[label={(\roman*)},noitemsep,topsep=0pt]
            \item \(a\) and \(b\) are comparable;
            \item \(u\) is increasing in \(\theta\) for \(a\) and \(b\);
            \item \(u(b_\theta) \geq (>) u(a_\theta)\) imply \(u(b_{\theta'}) \geq (>) u(a_{\theta'})\) for all \(\theta' \underline{\blacktriangleright} \theta\); and
            \item \(a \succeq_S b\).
        \end{enumerate}
    \end{definition}
    Then,
     \begin{remark}
        \(a \succeq_P b\) if and only if \(b \geq a\). Moreover, \(\succeq_P\) is a partial order on the set of actions.
    \end{remark}

\end{document}